\documentclass[12pt]{article}
\usepackage{amsmath,amssymb}
\usepackage{authblk}
\usepackage{color}
\usepackage{mathrsfs}
\usepackage{mathptmx}
\usepackage{verbatim}
\usepackage{epsfig}
\usepackage{CJK}
\usepackage{bm}
\usepackage{amsfonts}
\usepackage{amsthm}
\input{epsf.sty}
\usepackage{epsfig}
\usepackage{scalerel}
\usepackage{stackengine}

\def\<{\langle}
\def\>{\rangle}

\def\be{\begin{equation}}
\def\ee{\end{equation}}
\def\ba{\begin{array}}
\def\ea{\end{array}}

\newtheorem{theorem}{Theorem}[section]

\newtheorem{proposition}{Proposition}[section]
\newtheorem{corollary}{Corollary}[section]
\theoremstyle{definition}
\newtheorem{remark}{Remark}[section]
\newtheorem{example}{Example}[section]
\newtheorem{definition}{Definition}[section]
\numberwithin{equation}{section}

\newcommand\reallywidetilde[1]{\ThisStyle{%
  \setbox0=\hbox{$\SavedStyle#1$}%
  \stackengine{-.1\LMpt}{$\SavedStyle#1$}{%
    \stretchto{\scaleto{\SavedStyle\mkern.2mu\sim}{.5467\wd0}}{.7\ht0}%
  }{O}{c}{F}{T}{S}%
}}

\usepackage{amsmath}
\usepackage{amsthm}
\usepackage{amsfonts}
\usepackage{pifont}
\usepackage{color}
\usepackage{bbm}
\usepackage{CJK}
\usepackage{mathrsfs}
\usepackage{fancyhdr}
\usepackage{graphicx}
\usepackage{pdfsync}
\usepackage{psfrag}
\usepackage{time}
\usepackage{txfonts}
\usepackage{tikz-cd}

\usepackage{stmaryrd}
\usepackage{mathrsfs}
\usepackage{txfonts}
\usepackage{amsfonts}

\usepackage{indentfirst}

\def\be{\begin{equation}}
\def\ee{\end{equation}}
\def\br{\begin{eqnarray}}
\def\er{\end{eqnarray}}

\title{Multiscale method, Central extensions and a generalized Craik-Leibovich equation}

\author[1,2]{Cheng Yang\thanks{chyang@math.toronto.edu, c$\_$yang11@fudan.edu.cn}}
\affil[1]{Department of Mathematics, University of Toronto, Toronto, ON \newline M5S 2E4, Canada}
\affil[2]{School of Mathematical Sciences, Fudan University, Shanghai 200433,  \newline People's Republic of China}

\begin{document}

\maketitle

\date{}

\begin{quote}
\small {\bf Abstract} In this paper we develop perturbation theory on the reduced space of a principal $G-$bundle. This theory uses a multiscale method and is related to vibrodynamics. For a fast oscillating motion with the symmetry Lie group $G$, we prove that the averaged equation (i.e. the equation describing the slow motion) is the Euler equation on the dual of a certain central extension of the corresponding Lie algebra $\mathfrak g$. As an application of this theory we study the Craik--Leibovich (CL) equation in hydrodynamics. We show that CL equation can be regarded as the Euler equation on the dual of an appropriate central extension of the Lie algebra of divergence-free vector fields. From this geometric point of view, one can give a generalization of CL equation on any Riemannian manifold with boundary.\\
For accuracy of the averaged equation, we prove that the difference between the solution of the averaged equation and the solution of the perturbed equation remains small (of order $\epsilon$) over a very long time interval (of order $1/{\epsilon^2}$). Combining the geometric structure of the generalized CL equation and the averaging theorem, we present a large class of adiabatic invariants for the perturbation model of the Langmuir circulation in the ocean.
\end{quote}

\tableofcontents

\section{Introduction}

\par
In this paper we develop perturbation theory to study a fast oscillating system corresponding to a Lie group $G$. More specifically, we apply a fast-slow multiscale method to an oscillating Hamiltonian system on the reduced space of a principal $G-$bundle. It turns out that the averaged equation describing the slow motion is the Euler equation on the dual of the central extension of the corresponding Lie algebra $\mathfrak{g}$.
\par
We hope this theory can shed some lights on the geometric nature of some famous equations in mathematical physics, e.g. Craik-Leibovich equation for Langmuir circulation in oceans, infinite conductivity equation in plasma physics, $\beta-$plane or Rossby waves equation for a rotating fluid.
\par
This perturbation theory is related to vibrodynamics, an area of dynamics and hydrodynamics studying the behaviour of mechanical and fluid systems subject to fast oscillations. An interesting related example is the stability of the upper position of a pendulum with a vibrating suspension point. Vibrodynamics was studied by many authors including Kapitza, Landau, Bogolyubov, Yudovich, etc. A generalized Krylov-Bogolyubov averaging method related to two-timing procedure was studied by Yudovich, Vladimirov, etc and is a major tool in this area.
\par
Another interesting example in vibrodynamics is a flow with a fast oscillation related to the boundary conditions. By studying the oscillating flow, one can derive the Craik-Leibovich (CL) equation describing the Langmuir circulation in oceans. Recall that the CL equation is
\begin{equation}\label{e14}
\frac{\partial v}{\partial t}+(v,\nabla)v+curl\;v\times V_0=-\nabla p,
\end{equation}
where $V_0$ is a prescribed Stokes drift velocity. This equation was first studied by Craik and Leibovich in \cite{crle}. Vladimirov and his coauthors give a new derivation of Craik-Leibovich equation by using the generalized Krylov-Bogolyubov averaging method related to two-timing method, see \cite{vla}.
\par
In this paper we generalize this two-timing/averaging method to a perturbation theory on the principal $G-$bundle. By applying this theory to a principal $SDiff(D)-$bundle considered in \cite{lew} to describe the free boundary fluid motion, we derive the CL equation. This theory also leads us to the geometric meaning of the CL equation: it turns out to be the Euler equation on the dual of a certain central extension of the Lie algebra of divergence-free vector fields. This geometric point of view enables us to give a higher-dimensional generalization of the CL equation on any Riemannian manifold with boundary in any dimension. Also, a large class of invariant functionals follows from this geometric structure.
\par
Euler equations on the duals of central extensions of Lie algebras arise in many interesting settings in mathematical physics.
\begin{example}
 In \cite{khch} Khesin and Chekanov studied the infinite conductivity equation on a Riemannian manifold $M$:
\begin{equation}\label{e13}
\frac{\partial v}{\partial t}=-(v,\nabla)v-v\times B-\nabla p,
\end{equation}
where $B$ is a constant divergence-free magnetic field. This equation is the Euler equation on the dual space of the central extension of the Lie algebra of the divergence-free vector fields $SVect(M)$. The corresponding 2-cocycle is a Lichnerowicz 2-cocycle (see section 3.2) related to the magnetic field $B$:
$$
\widehat{\omega}_{B}(X,Y)=\int_{M}i_X i_Y i_B\nu,
$$
where $B$ is an $(n-2)-$vector field corresponding to a closed 2-form on $M$. Khesin and Chekanov generalized the infinite conductivity equation to any Riemannian manifolds in any dimension and found a large class of invariant functionals.
\end{example}
\begin{remark}
 Since the CL equation has a geometric structure similar to the infinite conductivity equation, we are able to prove that those invariants for the infinite conductivity equation turn out to be also invariants for the CL equation. This also helps construct a large class of adiabatic invariants for the fast-slow system related to the CL equation.
\end{remark}

\begin{example}
 In \cite{zeit} Zeitlin studied the $\beta$-plane equation (or Rossby waves equation):
$$
\dot{\omega}+\{\psi,\omega\}+\beta\psi_x=0,
$$
where $\beta$ is a constant related to the Coriolis force, $\omega$ and $\psi$ are the vorticity and stream functions, respectively. This equation describes the fluid motion on a rotating surface. It is the Euler equation on the dual of a central extension of the Lie algebra of the symplectomorphism group.
\end{example}

\par
Besides the geometric structure, the accuracy of the averaged equation is also considered in this paper. We prove the averaging theorem in a general setting. The averaging theorem combined with the geometric structure of the CL equation enables us to present a class of adiabatic invariants for the fast-slow system related to the CL equation.\\
\subsection*{Organization and main results of the paper}
\par
In section 2, we give the general setting of the perturbation theory. We derive the averaged equation for a perturbed ODE related to a bilinear operator on a Banach space. Here the main statement is the following averaging theorem:
\begin{theorem}{(=Theorem \ref{t01})}
 The difference between the solution of the averaged equation and the solution of the perturbed equation remains small (of order $\epsilon$) over a very long time interval (of order $\frac{1}{\epsilon^2}$).
\end{theorem}
\par
In section 3, we give some preliminaries about an Euler equation on the dual of a Lie algebra and a central extension of a Lie algebra.
\par
In section 4, we present a general theory on the reduced space of a principal $G-$bundle. We consider a natural fast-slow Hamiltonian system and derive the averaged equation. The Eulerian nature of this averaged equation is proved in theorem \ref{t21}:
\begin{theorem}{(=Theorem \ref{t21})}
 The averaged equation (i.e. the equation describing the slow motion) is the Euler equation on the dual of a central extension of the corresponding Lie algebra.
\end{theorem}
\noindent The accuracy of the averaged equation in this perturbation theory is described in theorem \ref{t22}. Using this averaging theorem, we study adiabatic invariants of this perturbation model.
\par
In section 5 we apply our theory to an incompressible fluid with a free boundary and obtain the CL equation. This yields the geometric structure of the CL equation and the Stokes drift: the CL equation can be seen as the Euler equation on the dual of a certain central extension of the Lie algebra of divergence-free vector fields. By using this geometric point of view we generalize the CL equation to any Riemannian manifolds with boundary.
\begin{theorem}{(=Theorem \ref{t31})}
 The $n-$dimensional CL equation is
\begin{equation}
 \frac{d}{ds}\;[u]=-\mathcal{L}_{v+V_0}\;[u],
\end{equation}
where $v+V_0\in SDiff(D)$, and $[u]=[v^b]$ is an element in the quotient $\Omega^1(D)/d\Omega^0(D)$, the regular dual to the Lie algebra $SVect(D)$.

\end{theorem}

\noindent This geometric structure also gives us a large class of invariants for the system described by the CL equation. Combining the averaging theorem, we prove that these invariants are actually the adiabatic invariants for a perturbation model of the Langmuir circulation.
\begin{theorem}{(=Theorem \ref{t36})}
 For the perturbation model of the CL equation on an $n-$dimensional Riemannian manifold $D$,
\par
 (1) the functional $I(v)=\int_D u\wedge\;(du)^m$ is an adiabatic invariant for 1-form $u=v^b$ for $n=2m+1$;
\par
 (2) the functionals $I_f(v)=\int_D f\left(\frac{(du)^m}{vol_D}\right)\;vol_D $ are adiabatic invariants for any function $f:\mathbb{R}\rightarrow\mathbb{R}$ and 1-form $u=v^b$ for $n=2m$.
\end{theorem}

\par
In appendix we summarize the geometry and local Poisson structure on the reduced space of a principal $G-$bundle \cite{mont} and their application to the study of an incompressible fluid with free boundary \cite{lew}. We also derive the equation for a local Hamiltonian system on the reduced space.
\section{Setting of the Perturbation theory}
\subsection{Derivation of the averaged equation}
We start with the following general setting for averaging. Let $X$ be a Banach space equipped with a norm $\|\cdot\|$, and $B:X\times X\rightarrow X$ is a bilinear operator. Consider an ordinary differential equation on $X$ of the form
\begin{equation}\label{e01}
\frac{d x}{d t}=B(x,y),
\end{equation}
where $x,y\in X$. We assume that both $x$ and $y$ have the following expansions in small parameter $\epsilon\rightarrow 0$:
$$
x=\epsilon^2 x_2+\epsilon^3 x_3+\epsilon^4 x_4+\dots,
$$
$$
y=\epsilon y_1+x=\epsilon y_1+\epsilon^2 x_2+\epsilon^3 x_3+\epsilon^4 x_4+\dots,
$$
where $y_1$ is a prescribed time-dependent vector periodic with respect to $t$. Now we introduce a slow time $s=\epsilon^2 t$, and for functions depending on both fast and slow times the time derivative becomes $\frac{d}{dt}=\frac{\partial}{\partial t}+\epsilon^2\frac{\partial}{\partial s}$.
\par
Now assume that $x=x(s,t)$ is such a function of both times, so $y=y(s,t)=\epsilon y_1(t)+x(s,t)$, and equation (\ref{e01}) becomes
\begin{equation}\label{e02}
\frac{\partial x}{\partial t}+\epsilon^2\frac{\partial x}{\partial s}=B(x,y).
\end{equation}
\par
Now we are going to derive the corresponding averaged equation.
\begin{proposition}\label{p01}
 The averaged equation for the equation (\ref{e01}) is
\begin{equation}\label{e03}
\frac{d}{ds}\;\bar{x}_2=\overline{B(B(\bar{x}_2,y_1^t),y_1)}+B(\bar{x}_2,\bar{x}_2),
\end{equation}
or, equivalently,
\begin{equation}\label{e04}
\frac{d}{dt}\;\bar{x}_2=\epsilon^2 \left(\overline{B(B(\bar{x}_2,y_1^t),y_1)}+B(\bar{x}_2,\bar{x}_2)\right),
\end{equation}
where the slow time is $s=\epsilon^2 t$ and
\begin{align*}
y_1^{t}&:=\int_0^{t}y_1(\sigma)\;d\sigma-\overline{\int_0^{t}y_1(\sigma)\;d\sigma}\\
&=\int_0^{t}y_1(\sigma)\;d\sigma-\frac{1}{2\pi}\int_0^{2\pi}\int_0^{\mu}y_1(\sigma)\;d\sigma\; d\mu.
\end{align*}
\end{proposition}
\begin{proof}
Split $x(s,t)$ into the average and oscillating parts (denoted respectively by straight and wave overlines),
$$
x(s,t)=\bar{x}(s)+\tilde{x}(s,t),
$$
where
$$
\bar{x}(s)=\frac{1}{2\pi}\int_0^{2\pi}x(s,t)\;dt,\;\;\text{and}\;\;\tilde{x}(s,t)=x(s,t)-\bar{x}(s).
$$
Equating the coefficients at $\epsilon^2,\;\epsilon^3,\;\epsilon^4$ to zero in equation (\ref{e02}), we obtain
$$
\epsilon^2:\;\;\tilde{x}_{2t}=0\Rightarrow\tilde{x}_2=0,
$$
\begin{equation*}
\epsilon^3:\;\;\frac{\partial\tilde{x}_3}{\partial t}=B(\bar{x}_2,y_1)\Rightarrow\tilde{x}_3=B(\bar{x}_2,y_1^t),
\end{equation*}
where
$$
y_1^{t}=\int_0^{t}y_1(\sigma)\;d\sigma-\overline{\int_0^{t}y_1(\sigma)\;d\sigma}.
$$
Equating the coefficients at $\epsilon^4$ to zero in equation (\ref{e02}), we get
$$
\epsilon^4:\;\;\frac{\partial\tilde{x}_4}{\partial t}+\frac{\partial\bar{x}_2}{\partial s}=B(x_3,y_1)+B(x_2,x_2),
$$
Upon averaging this equation with respect to $t$, we have
$$
\frac{d}{ds}\;\bar{x}_2=\overline{B(B(\bar{x}_2,y_1^t),y_1)}+B(\bar{x}_2,\bar{x}_2),
$$
as required.
\end{proof}
It turns out that if the bilinear operator $B$ satisfies certain properties, the averaged equation can be written as an equation with a shift term,
\begin{proposition}\label{p02}
 If the linear operator $B$ is antisymmetric and satisfies the Jacobi identity, then the averaged equation (\ref{e03}) becomes
$$
\frac{d}{ds}\;\bar{x}_2=B(-\frac 12 \overline{B(y_1,y_1^t)}+\bar{x}_2,\bar{x}_2),
$$
where $s=\epsilon^2 t$.
\end{proposition}
\begin{proof}
 By the Jacobi identity
\begin{align*}
B(B(y_1^t,y_1),\bar{x}_2)&=-B(B(\bar{x}_2,y_1^t),y_1)-B(B(y_1,\bar{x}_2),y_1^t)\\
&=B(B(\bar{x}_2,y_1),y_1^t)-B(B(\bar{x}_2,y_1^t),y_1),
\end{align*}
where the last equality is due to antisymmetry of operator $B$.\\
On the other hand,
$$
0=\overline{\frac{d}{dt}B(B(\bar{x}_2,y_1^t),y_1^t)}=\overline{B(B(\bar{x}_2,y_1),y_1^t)}+\overline{B(B(\bar{x}_2,y_1^t),y_1)}.
$$
So we obtain
$$
\overline{B(B(\bar{x}_2,y_1^t),y_1)}=-B(\frac 12 \overline{B(y_1,y_1^t)},\bar{x}_2),
$$
as required.
\end{proof}

\subsection{The averaging theorem}

Let $\mathfrak{D}$ be a bounded domain in Banach space $X$ with norm $\|\cdot\|$ and assume that the solution $\bar{x}_2(t)$ of the averaged equation (\ref{e04}) with initial condition $\bar{x}_2(0)=a_{initial}\in \mathfrak{D}$ remains in $\mathfrak{D}$ over a time of $\frac{T}{\epsilon^2}$ for sufficiently small $\epsilon$. (In other words, the solution of averaged equation (\ref{e03}) with the same initial condition stays in the domain $\mathfrak{D}$ over time $T$.)
Then we have the following theorem:
\begin{theorem}\label{t01}
 Let $x(t)$ be the solution of equation (\ref{e01}) with initial condition $x(0)=\epsilon^2\bar{x}_2(0)=\epsilon^2 a_{initial}$. The difference between the value of the solution $\bar{x}_2(t)$ of the averaged equation (\ref{e04}) and $x^*:=\frac{x}{\epsilon^2}$ remains small for $t\in[0,\frac{T}{\epsilon^2}]$ if $\epsilon$ is sufficiently small:
$$
\|x^*(t)-\bar{x}_2(t)\|\leq C\epsilon,
$$
where the constant $C$ is independent of $\epsilon$.

\end{theorem}

\begin{proof}
 First we choose a new coordinate $x_2^*$ in such a way that the old coordinate $x^*$ can be expressed via $x_2^*$ in the form
\begin{equation}\label{e05}
 x^*=x_2^*+\epsilon x_3^*(x_2^*,t)+\epsilon^2 x_4^*(x_2^*,t)+\dots,
\end{equation}
where all $x_i^*(x_2^*,t),\;i\geq 3$ have zero mean with respect to $t$.
\par
We know that $x=\epsilon^2 x^*$ and $y=\epsilon y_1(t)+x$, so plug this into (\ref{e01}) to obtain
\begin{multline}\label{e06}
\epsilon^2\frac{dx_2^*}{dt}+\epsilon^3\frac{\partial x_3^*}{\partial t}+\epsilon^3\frac{\partial x_3^*}{\partial x_2^*}\frac{dx_2^*}{dt}+\epsilon^4\frac{\partial x_4^*}{\partial t}+\epsilon^4\frac{\partial x_4^*}{\partial x_2^*}\frac{dx_2^*}{dt}+\dots\\
=\epsilon^3 B(x_2^*,y_1)+\epsilon^4 B(x_2^*,x_2^*)+\epsilon^4 B(x_3^*,y_1)+\dots.
\end{multline}
Now we can assume that the derivative of $x_2^*$ has an expansion
\begin{equation}\label{e07}
 \frac{dx_2^*}{dt}=\epsilon^2 M_2(x_2^*)+\epsilon^3 M_3(x_2^*)+\dots,
\end{equation}
because $x_2^*$ is the slow variable which only depend on slow time $s=\epsilon^2t$, we have $\frac{dx_2^*}{dt}=\epsilon^2 \frac{dx_2^*}{ds}\sim O(\epsilon^2)$, so there is no $\epsilon$ term in the above expansion.\\
\par
Then substituting (\ref{e07}) into (\ref{e06}), we get
\begin{equation}\label{e08}
 \epsilon^3\left(\frac{\partial x_3^*}{\partial t}-B(x_2^*,y_1)\right)+\epsilon^4\left(\frac{\partial x_4^*}{\partial t}+M_2(x_2^*)-B(x_2^*,x_2^*)-B(x_3^*,y_1)\right)+O(\epsilon^5)=0.
\end{equation}
Now we assume that $x_3^*$ is of the form
$$
x^*_3=B(x^*_2,y_1^t),
$$
and
$$
x^*_4=-\int_{0}^{t}\left(\reallywidetilde{B(x_3^*,y_1)}+B(x_2^*,x_2^*)\right)\;d\tau+\overline{\int_{0}^{t}\left(\reallywidetilde{B(x_3^*,y_1)}+B(x_2^*,x_2^*)\right)\;d\tau},
$$
so we obtain
$$
\frac{\partial x_3^*}{\partial t}=B(x_2^*,y_1)
$$
and
$$
\frac{\partial x_4^*}{\partial t}=\reallywidetilde{B(x_3^*,y_1)}+B(x_2^*,x_2^*).
$$
We also assume that
$$
M_2(x_2^*)=B(x_2^*,x_2^*)+\overline{B(x_3^*,y_1)},
$$
plug these into equation (\ref{e08}), we have
$$
\frac{d}{ds}\;x_2^*=\overline{B(B(x_2^*,y_1^t),y_1)}+B(x_2^*,x_2^*)+\epsilon R(x_2^*,\epsilon),
$$
where $\|R\|\leq C_1$ provided that $x^*,\;x_2^*$ belong to $\mathfrak{D}$.
\par
Comparing this with the averaged equation (\ref{e03}), we obtain that $z=x_2^*-x^*$ satisfies the inequality
$$
\frac{d}{ds}\|z\|\leq a \|z\| + b,
$$
where  $a$ is a constant and $b=c_1\epsilon$ as long as $x^*,\;x_2^*,\;\bar{x}_2$ remain in $\mathfrak{D}$. From this differential inequality, we obtain the estimate
$$
\|z\|\leq bse^{as},
$$
as long as $x^*,\;x_2^*,\;\bar{x}_2$ remain in $\mathfrak{D}$. For a finite time $T$, this yields the estimate
$$
\|x_1^*(t)-\bar{x}_1(t)\|\leq C_2\epsilon, \;\;C_2=C_1Te^{aT}.
$$
On the other hand, we have
$$
\|x^*(t)-x_2^*(t)\|\leq C_3\epsilon.
$$
Let $d$ be the distance from the trajectory of the averaged motion $\{x_1(t),\;t\leq\frac{T}{\epsilon^2}\}$ to the boundary of $\mathfrak{D}$. Choose $\epsilon$ so that $(C_2+C_3)\epsilon\leq d$, then $x^*,\;x_2^*,\;\bar{x}_2$ remain in $\mathfrak{D}$ for $t\in[0,\frac{T}{\epsilon^2}]$. So we get
$$
\|x^*(t)-\bar{x}_2(t)\|\leq \|x^*(t)-x_2^*(t)\|+\|x_2^*(t)-\bar{x}_2(t)\|\leq C\epsilon, \;\;C=C_2+C_3.
$$
\end{proof}
\begin{remark}
 The proof of theorem \ref{t01} manifests that the two-timing method (when one considers slow and fast times as independent variables) and averaging method are equivalent for nonlinear oscillating system (\ref{e01}). This allows us to give the rigorous justification of this two-timing method. As an application of this perturbation theory, we are going to study the free boundary problem of an incompressible fluid. The above consideration also allows one to justify various formal applications of the two-timing method to PDE, c.f. \cite{crle}, \cite{vla}.
\end{remark}
\begin{remark}
 Below we are also going to obtain the averaged equation of a perturbed Hamiltonian system on the reduced space of a principal $G$-bundle. The Hamiltonian equation studied below is a special case of (\ref{e01}), so the solution of the averaged equation well approximates the actual one by theorem \ref{t01}.
\end{remark}

\section{Geometric Preliminaries}

\subsection{Euler equations on the dual of Lie algebras}
\par
Let $G$ be a finite or infinite-dimensional Lie group, $\mathfrak{g}$ its Lie algebra, and $\mathfrak{g}^*$ stands for (the regular part of) its dual.
\begin{definition}\label{d11}
 The \textbf{natural Lie--Poisson structure} $\{\;,\;\}_{LP}:C^{\infty}(\mathfrak{g}^*)\times C^{\infty}(\mathfrak{g}^*)\rightarrow C^{\infty}(\mathfrak{g}^*)\;$ on the dual space $\mathfrak{g}^*$ is the Poisson bracket defined by
$$
\{f,g\}_{LP}(m):=\langle [df,dg], m\rangle
$$
for any $m\in \mathfrak g^*$ and smooth functions
$f,g$ on $\mathfrak g^*$. Here the differentials are taken at the point $m$, and
$\langle \cdot,\cdot\rangle$ is a natural pairing between $\mathfrak g$
and $\mathfrak g^*$.
\end{definition}
\begin{proposition}\label{p11}
 The Hamiltonian equation corresponding to a function $H$ and the Lie--Poisson structure $\{\;,\;\}_{LP}$ on $\mathfrak{g}^*$ is given by
\begin{equation}\label{e11}
\frac{dm}{dt}=ad^*_{dH}m.
\end{equation}
\end{proposition}
For proof, see e.g. \cite{khmi}.

\begin{definition}\label{d12}
 The \textbf{Euler equation} on $\mathfrak{g}^*$ is an equation corresponding to the quadratic (energy) Hamiltonian
$H(m)=-\frac 12\langle \mathbb{I}^{-1}m,m\rangle$:
\begin{equation}\label{e12}
\frac{dm}{dt}=-ad^*_{\mathbb{I}^{-1}m}m,
\end{equation}
where $\mathbb{I}:\mathfrak{g}\rightarrow \mathfrak{g}^*$ is an inertia operator.
\end{definition}
\begin{remark}
 Arnold in \cite{arn} developed the general theory for the Euler equation describing the geodesic flow on an arbitrary Lie group equipped with a one-sided invariant metric. Given a Lie group $G$, consider the right-invariant metric on $G$ which is defined at the group identity by the quadratic form corresponding to an inertia operator $\mathbb{I}:\mathfrak{g}\rightarrow \mathfrak{g}^*$. Arnold proved that for such a right-invariant metric on group, the corresponding geodesic flow is described by the equation (\ref{e12}). This equation (\ref{e12}) coincides with the classical Euler equation of an ideal fluid in the case of the group $G=SDiff(M)$ and the right-invariant $L^2$-metric. The case of a rigid body with a fixed point is related to the group $G=SO(3)$ and a left-invariant metric, and the corresponding equation differs by sign from (\ref{e12}). More detailed discussion can be found in \cite{arkh}.
\end{remark}

\subsection{Central extensions}
\begin{definition}\label{d13}
 A \textbf{central extension} of a Lie algebra $\mathfrak{g}$ by a vector space $V$ is a Lie algebra $\hat{\mathfrak{g}}$ whose underlying vector space $\hat{\mathfrak{g}}=\mathfrak{g}\oplus V$ is equipped with the Lie bracket:
$$
[(X,u),(Y,v)]^{\wedge}=([X,Y],\;\widehat{\omega}(X,Y)),
$$
for a Lie algebra 2-cocycle $\widehat{\omega}:\mathfrak{g}\times\mathfrak{g}\rightarrow V$, i.e. for a bilinear, antisymmetric form $\widehat{\omega}$ on the Lie algebra that satisfies the cocycle identity:
$$
\widehat{\omega}([X,Y],Z)+\widehat{\omega}([Y,Z],X)+\widehat{\omega}([Z,X],Y)=0.
$$
\end{definition}
Next, we are going to define a special 2-cocycle, which appears in section 4 when studying the averaged equation.
\begin{definition}\label{d14}
 Fix a vector $V_0\in\mathfrak{g}$ and define an \textbf{averaging 2-cocycle} $\widehat{\omega}_{V_0}:\mathfrak{g}\times\mathfrak{g}\rightarrow \mathbb{R}$ on the Lie algebra $\mathfrak{g}$ by
\begin{equation}\label{e113}
 \widehat{\omega}_{V_0}(X,Y)=\left\langle ad^*_{V_0}\;\mathbb{I}(X),Y\right\rangle,
\end{equation}
where $X,\;Y\in\mathfrak{g}$ and $\mathbb{I}$ is the inertia operator on $\mathfrak{g}$.
\end{definition}
\begin{remark}
 Note that $\widehat{\omega}_{V_0}$ is a trivial 2-cocycle, or 2-coboundary, since $\left\langle ad^*_{V_0}\;\mathbb{I}(X),Y\right\rangle=-\left\langle \mathbb{I}(V_0),[X,Y]\right\rangle$.
\end{remark}

\par
We have the following theorem about the Euler equation on $\hat{\mathfrak{g}}_{V_0}^*$,
\begin{theorem}\label{t11}
 Let $\hat{\mathfrak{g}}_{V_0}$ be the central extension of the Lie algebra $\mathfrak{g}$ with the 2-cocycle $\widehat{\omega}_{V_0}$. Then the Euler equation on $\hat{\mathfrak{g}}_{V_0}^*$ corresponding to the quadratic (energy) Hamiltonian $H(m)=-\frac{1}{2}\langle\mathbb{I}^{-1}m,m\rangle$ is
\begin{equation}\label{e114}
 \frac{d}{dt}\;m=-ad^*_{\mathbb{I}^{-1}m+V_0}\;m.
\end{equation}

\end{theorem}
\begin{proof}
 Since
\begin{align*}
 \left\langle ad^*_{(X,a)}(m,b),(Y,c)\right\rangle
=&\left\langle (m,b),([X,Y],\widehat{\omega}_{V_0}(X,Y))\right\rangle
=\left\langle m,[X,Y]\right\rangle+b\;\widehat{\omega}_{V_0}(X,Y)\\
=&\left\langle ad^*_X\;m,Y\right\rangle+\left\langle b \;ad^*_{V_0}\mathbb{I}(X),Y\right\rangle
=\left\langle ad^*_X\;m+b \;ad^*_{V_0}\mathbb{I}(X),Y\right\rangle,
\end{align*}
we get that the Euler equation on the dual space $\hat{\mathfrak{g}}_{V_0}^*$ of the central extension of $\mathfrak{g}$ (for $b=1$) is
$$
\frac{d}{dt}\;m=-ad^*_{\mathbb{I}^{-1}m+V_0}\;m.
$$

\end{proof}

\begin{remark}
Let $M$ be a compact manifold with a volume form $\nu$ and $\beta$ is a closed 2-form on $M$. The Lichnerowicz 2-cocycle $\widehat{\omega}_{\beta}$ on Lie algebra $SVect(M)$ of divergence-free vector fields on $M$ (tangent to the boundary of $M$) is defined by
$$
\widehat{\omega}_{\beta}(X,Y)=\int_{M}\beta(X,Y)\nu.
$$
The infinite conductivity equation is the Euler equation on the dual space of the central extension of $SVect(M)$ \cite{khch}. The 2-form $\beta$ plays the role of a magnetic field on $M$.
\end{remark}

\section{Averaging and Lie groups }

\subsection{Perturbation theory on a principle bundle}
\par
Let $\pi:M\rightarrow N$ be a principle $G-$bundle and $\widetilde{M}=\pi^*(T^*N)$ is the pullback of the cotangent bundle $T^*N$. Consider a natural Hamiltonian system on the reduced space $S:=\widetilde{M}\times_G \mathfrak{g}^*$. Locally, this space is isomorphic to $T^*U\times \mathfrak{g}^*$, where $U$ is an open subset of $N$. In the local coordinates, let the Hamiltonian function be $H(q,\;p,\;\mu)=\frac 12 \|p\|^2+\frac 12 \langle\mu,\mathbb{I}^{-1}\mu\rangle+V(q)$, where $\|\;\|$ is the norm on $T^*U$ induced from Riemannian metric on $N$. More detailed discussion of the geometry and Hamiltonian structure on this reduced space can be found in the appendix.
\par
In the appendix (see proposition \ref{p21}) we show that the equations of the natural Hamiltonian system are
\begin{equation}\label{e22}
\left\{
\begin{aligned}
 \dot{\mu}&=\;-ad^*_{\mathbb{I}^{-1}\mu-\widetilde{A}_q\frac{\delta H}{\delta p}}\;\mu,\\
 \dot{p}&=\;-\nabla V(q)+\widetilde{A}^*_q ad^*_{\mathbb{I}^{-1}\mu}\;\mu-\widetilde{\Omega}_{q,\frac{\delta H}{\delta p}}^*\;\mu,\\
 \dot{q}&=\;\frac{\delta H}{\delta p}.
\end{aligned}
\right.
\end{equation}
Here, $\frac{\delta H}{\delta p}\in TU$ is a velocity field on the subset $U$, while operators $\widetilde{A}_q:T_q U\rightarrow \mathfrak{g}$ and  $\widetilde{\Omega}_{q,\frac{\delta H}{\delta p}}=\widetilde{\Omega}_q\left(\frac{\delta H}{\delta p},\cdot\right):T_q U\rightarrow \mathfrak{g}$, as well as their duals $\widetilde{A}_q^*:\mathfrak{g}^*\rightarrow T^*_q U$ and $\widetilde{\Omega}_{q,\frac{\delta H}{\delta p}}^*:\mathfrak{g}^*\rightarrow T^*_q U$, are certain operators whose geometric meanings will be explained in the appendix.
\begin{remark}
 In this setting, $\mu$ is the slow variable and $p,\;q$ are the fast variables. The first equation of Equations (\ref{e22}) has the form of equation (\ref{e01}), and one can apply to it the two-timing method discussed in section 2.1.
\end{remark}

\par
Now, we apply two-timing method to the Hamiltonian system (\ref{e22}). In our perturbation theory, $q(t)$ belongs to a small neighbourhood of an averaged position $\bar{q}$. So, we can use the local equation in proposition \ref{p21}.
\par
We look for the solution of the form $(\mu,\;p,\;q)(s,\tau)$, where $\mu,\;p,\;q$ are functions of time $t$ and perturbation parameter $\epsilon$, but, in fact, they are functions of two time variables $s,\;\tau$. Note, however, that now we introduce different fast and slow times: $\tau=\frac{t}{\epsilon}$ is the new fast time, and $s=\epsilon t$ is the new slow time, and $\tau-$dependence is $2\pi-$periodic, but $s-$dependence is not necessarily periodic.
\par
By the chain rule, we find that the first equation of (\ref{e22}) becomes
\begin{equation}\label{e23}
\frac{1}{\epsilon}\frac{\partial}{\partial\tau}\mu+\epsilon\frac{\partial}{\partial s}\mu=\;-ad^*_{\mathbb{I}^{-1}\mu-\widetilde{A}_q\frac{\delta H}{\delta p}}\;\mu.
\end{equation}
Expand $\mu(s,\tau)$ and note that the indices are shifted compared with the ones used above in section 2 since the slow and fast times are defined differently.
\begin{equation}\label{e24}
\mu=\mu_0(s,\tau)+\epsilon\mu_1(s,\tau)+\epsilon^2\mu_2(s,\tau)+O(\epsilon^3).
\end{equation}
We also split $\mu$ into the average and oscillating parts,
$$
\mu(s,\tau)=\bar{\mu}(s)+\tilde{\mu}(s,\tau),
$$
where $\bar{\mu}$ is the average part of $\mu(s,\tau)$ w.r.t. $\tau$, i.e. $\bar{\mu}=\frac{1}{2\pi}\int_0^{2\pi}\mu(s,\tau)d\tau$.
\par
Likewise, we expand $v= \mathbb{I}^{-1}\mu-\widetilde{A}_q\frac{\delta H}{\delta p}$,
\begin{equation}\label{e25}
v=v_0(s,\tau)+\epsilon v_1(s,\tau)+\epsilon^2 v_2(s,\tau)+O(\epsilon^3),
\end{equation}
and split it into the average and oscillating parts, too.
We are making the following 3 assumptions:
\par
1. $\bar{\mu}_0=0,\;\bar{v}_0=0$, which means that the zeroth approximation of the average motion is zero.
\par
2. $\tilde{v}_0$ does not depend on $s$, i.e. the zeroth approximation is a purely fast motion.
\par
3. $\nabla V(q)|_{q=\bar{q}}=0$, which means that the force field of the potential has an equilibrium at the average position $\bar{q}$. \\
\begin{remark}
 As we will see below, assumption 1 and 2 together guarantee that the zeroth approximation of the velocity field generates a purely oscillating potential flow. We call it a potential flow, because in fluid dynamics this velocity field corresponds to an irrotational potential flow.
\end{remark}

\begin{theorem}\label{t21}
 The first approximation of the slow (averaged) motion is described by the following Hamiltonian equation on $\mathfrak{g}^*$:
\begin{equation}\label{e28}
 \frac{d}{ds}\;\bar{\mu}_1=-ad^*_{\bar{v}_1+V_0}\;\bar{\mu}_1,
\end{equation}
where $V_0=\overline{\frac 12[\tilde{v}_0,\tilde{v}_0^{\tau}]}$. Moreover, the equation (\ref{e28}) is the Euler equation on the dual space $\hat{\mathfrak{g}}_{V_0}^*$ of the central extension of $\mathfrak{g}$ with the averaging 2-cocycle
\begin{equation}\label{e29}
 \widehat{\omega}_{V_0}(X,Y)=\left\langle ad^*_{V_0}\;\mathbb{I}(X),Y\right\rangle.
\end{equation}

\end{theorem}
\begin{proof}
The second statement follows from theorem \ref{t11}, since equation (\ref{e28}) is the Euler equation on the dual space of the central extension of $\mathfrak{g}$. Now we are going to derive this equation (\ref{e28}).
\par
Equating the coefficients at $\frac{1}{\epsilon},\;\epsilon^0,$ and $\epsilon$ to zero in equation (\ref{e23}), we obtain
$$
\frac{1}{\epsilon}:\;\;\tilde{\mu}_{0\tau}=0\Rightarrow\tilde{\mu}_0=0,
$$
$$
\epsilon^0:\;\;\tilde{\mu}_{1\tau}=0\Rightarrow\tilde{\mu}_1=0,
$$
\begin{equation}\label{e26}
\epsilon:\;\;\frac{\partial\tilde{\mu}_2}{\partial\tau}=-ad^*_{\tilde{v}_0}\;\bar{\mu}_1\Rightarrow\tilde{\mu}_2=-ad^*_{\tilde{v}_0^{\tau}}\;\bar{\mu}_1,
\end{equation}
where
\begin{align*}
\tilde{v}_0^{\tau}&=\int_0^{\tau}\tilde{v}_0(\sigma)d\sigma-\overline{\int_0^{\tau}\tilde{v}_0(\sigma)d\sigma}=\int_0^{\tau}\tilde{v}_0(\sigma)d\sigma-\frac{1}{2\pi}\int_0^{2\pi}\int_0^{\mu}\tilde{v}_0(\sigma)d\sigma d\mu.
\end{align*}
Equating the coefficients at $\epsilon^2$ to zero, we get
$$
\epsilon^2:\;\;\frac{\partial\tilde{\mu}_3}{\partial\tau}+\frac{\partial\bar{\mu}_1}{\partial s}=-ad^*_{\tilde{v}_0}\;\mu_2-ad^*_{v_1}\;\bar{\mu}_1.
$$
By averaging this equation w.r.t. $\tau$, we obtain
$$
\frac{d}{ds}\;\bar{\mu}_1=-ad^*_{\bar{v}_1}\;\bar{\mu}_1-\overline{ad^*_{\tilde{v}_0}\;\mu_2}.
$$
Plugging (\ref{e26}) into this, we get
\begin{equation}\label{e27}
 \frac{d}{ds}\;\bar{\mu}_1=-ad^*_{\bar{v}_1}\;\bar{\mu}_1-\overline{ad^*_{\tilde{v}_0}\;ad^*_{\tilde{v}_0^{\tau}}\;\bar{\mu}_1}.
\end{equation}
Note that we have
$$
\overline{ad^*_{\tilde{v}_0}\;ad^*_{\tilde{v}_0^{\tau}}}=ad^*_{\overline{\frac 12[\tilde{v}_0,\tilde{v}_0^{\tau}]}},
$$
because
$$
ad^*_{\tilde{v}_0}\;ad^*_{\tilde{v}_0^{\tau}}-ad^*_{\tilde{v}_0^{\tau}}\;ad^*_{\tilde{v}_0}=ad^*_{[\tilde{v}_0,\tilde{v}_0^{\tau}]}
$$
and
$$
0=\overline{\frac{d}{d\tau}ad^*_{\tilde{v}_0^{\tau}}\;ad^*_{\tilde{v}_0^{\tau}}}=\overline{ad^*_{\tilde{v}_0}\;ad^*_{\tilde{v}_0^{\tau}}+ad^*_{\tilde{v}_0^{\tau}}\;ad^*_{\tilde{v}_0}},
$$

\par
Now we set $V_0=\overline{\frac 12[\tilde{v}_0,\tilde{v}_0^{\tau}]}$, then the equation (\ref{e27}) becomes
$$
 \frac{d}{ds}\;\bar{\mu}_1=-ad^*_{\bar{v}_1+V_0}\;\bar{\mu}_1.
$$
\par
We claim that this is the only equation of the system (\ref{e22}) which has nontrivial averaging.
\par
Indeed, consider the second equation of (\ref{e22}). Since $\nabla V(q)|_{q=q_0}=0$ then after averaging with respect to $\tau$, the last 2 terms in RHS of the second equation of (\ref{e22}) are of the order 2, and we obtain
$$
\frac{d\bar{p}_1}{ds}=0.
$$
\par
Furthermore, for the third equation of (\ref{e22}), due to $\mu_0=0$, we have $\tilde{v}_0(\tau)=\widetilde{A}_q^*\left(\frac{\delta H}{\delta p}\right)_0$, where $\left(\frac{\delta H}{\delta p}\right)_0$ is the zeroth approximation of $\frac{\delta H}{\delta p}$. Since $q$ is in an $\epsilon$-neighbourhood of $\bar{q}$, equating the coefficients at $\epsilon$ in the third equation of (\ref{e22}), we get $\overline{\left(\frac{\delta H}{\delta p}\right)}_1=0$. So, the first approximation of the third equation is
$$
\frac{d\bar{q}_1}{ds}=0.
$$

Thus, the first approximation of the slow (averaged) motion is described by Hamiltonian equation (\ref{e28}) on $\mathfrak{g}^*$.

\end{proof}

\begin{remark}
 In the proof, we choose the slow time to be $\epsilon t$ and the fast time to be $\frac{t}{\epsilon}$. This two-timing method to derive the averaged equation is slightly different from the method in section 2.1 (where we chose the slow time to be $\epsilon^2 t$ and the fast time to be $t$ ). We will see in the next section that these two derivations are actually equivalent.
\end{remark}

\begin{corollary}\label{t210}
 Equation (\ref{e28}) can be also regarded as the Hamiltonian equation on the coadjoint orbits in $\mathfrak{g}^*$ with the shifted Hamiltonian function $H(\bar{\mu}_1)=-\frac{1}{2}\langle\bar{\mu}_1+\mathbb{I}V_0,\;\mathbb{I}^{-1}\bar{\mu}_1+V_0\rangle$.
\end{corollary}
\begin{proof}
By proposition (\ref{p11}), the Hamiltonian equation for the Hamiltonian function $H(\bar{\mu}_1)=-\frac{1}{2}\langle\bar{\mu}_1+\mathbb{I}V_0,\;\mathbb{I}^{-1}\bar{\mu}_1+V_0\rangle$ is
$$
\frac{d}{ds}\;\bar{\mu}_1=-ad^*_{\mathbb{I}^{-1}\bar{\mu}_1+V_0}\;\bar{\mu}_1,
$$
which coincides with the equation (\ref{e28}).
\end{proof}

\subsection{Accuracy of averaging on Lie groups}
\par
In this section, we present a theorem on the accuracy of the averaged equation (\ref{e28}) over time intervals of order $\frac{1}{\epsilon^2}$. (The classical result on averaging theorem can be found in \cite{arn2}.)
\par
First, let us consider a perturbation model equivalent to the perturbation theory developed earlier in section 4.1, so that we could apply the general perturbation theory in section 2.
\par
Assume that the velocity field $v=\mathbb{I}^{-1}\mu-\widetilde{A}_q\frac{\delta H}{\delta p}$ has the form
$$
v=\epsilon v_1+\epsilon^2 v_2+\epsilon^3 v_3+\dots,
$$
where, $v_1$ is a velocity field such that $v_1=-\widetilde{A}_q\frac{\delta H}{\delta p}$ and periodic with respect to time t.
Furthermore, for the variable $\mu$, we have
$$
\mu=\epsilon^2 \mu_2+\epsilon^3 \mu_3+\epsilon^4 \mu_4+\dots.
$$
So the first equation of (\ref{e22}) has the same form as equation(\ref{e01}), where the bilinear operator $B(x,y)=-ad^*_{\mathbb{I}^{-1}y} x$.
\par
Now consider  a slow time $s=\epsilon^2 t$, and let $v=v(s,t),\;\mu=\mu(s,t)$. Split these functions into the average and oscillating parts,
$$
v(s,t)=\bar{v}(s)+\tilde{v}(s,t),
$$
$$
\mu(s,t)=\bar{\mu}(s)+\tilde{\mu}(s,t).
$$
By proposition \ref{p01} and \ref{p02}, we have
$$
\frac{d}{ds}\;\bar{\mu}_2=-ad^*_{\bar{v}_2+V_1}\;\bar{\mu}_2
$$
or, equivalently,
\begin{equation}\label{e210}
 \frac{d}{dt}\;\bar{\mu}_2=-\epsilon^2 ad^*_{\bar{v}_2+V_1}\;\bar{\mu}_2,
\end{equation}
where $V_1=\overline{\frac 12[\tilde{v}_1,\tilde{v}_1^{\tau}]}$.
\begin{remark}
  Note that the indices in the averaged equations are shifted due to the different choices of slow and fast times.
\end{remark}

\par
Now let $\mathfrak{D}$ be a bounded domain in $\mathfrak{g}^*$ and assume that the solution $\mu_2(t)$ of the averaged equation (\ref{e210}) with initial condition $\mu_2(0)=m_{initial}\in \mathfrak{D}$ remains in $\mathfrak{D}$ over a time of $\frac{T}{\epsilon^2}$. Then we have the following theorem, which follows from theorem \ref{t01}:
\begin{theorem}\label{t22}
 Let $\mu(t)$ be the solution of first equation in (\ref{e22}) with initial condition $\mu(0)=\bar{\mu}_2(0)=m_{initial}$. The difference between the value of the solution $\bar{\mu}_2(t)$ of the average equation (\ref{e210}) and $\mu^*=\frac{\mu}{\epsilon^2}$ remains small for $t\in[0,\frac{T}{\epsilon^2}]$ provided that $\epsilon$ is sufficiently small:
$$
\|\mu^*(t)-\bar{\mu}_2(t)\|\leq C\epsilon,
$$
where the constant $C$ is independent of $\epsilon$ and the norm $\|\cdot\|$ is defined by $\|\mu\|=\langle\mu,\mathbb{I}^{-1}\mu\rangle^{\frac{1}{2}}$.

\end{theorem}

\subsection{Adiabatic invariants}
First, let us recall the definition of adiabatic invariants. Consider a Hamiltonian system whose parameters change slowly. Suppose that the Hamiltonian is $H=H(p,\;q,\;\lambda)$, where $\lambda=\lambda(\tau)$, $\tau=\delta t,\;0<\delta\ll 1$, and $\lambda(\tau)$ is assumed to be sufficiently smooth.
\begin{definition}[\cite{arko}]\label{d21}
A function $I(p,q,\lambda)$ is called an \textbf{adiabatic invariant} if for any $\kappa>0$, there exists $\delta_0=\delta_0(\kappa)$ such that for $\delta<\delta_0$, the change of $I(p(t),q(t),\lambda(\delta t))$ for $0\leq t\leq \frac{1}{\delta}$ does not exceed $\kappa$.
\end{definition}
\begin{theorem}\label{t24}
 The shifted energy $E=\frac 12 \langle\mu+\mathbb{I}V_0,\mathbb{I}^{-1}\mu+V_0\rangle$ is an adiabatic invariant for the perturbed Hamiltonian system studied in section 3.2. In other words, let $\mathfrak{D}$ be a bounded domain in the energy norm, then there exists $\epsilon_0$ such that for all $\epsilon$ satisfying $0<\epsilon<\epsilon_0$, one has
$$
\mid E(\mu(t),p(t),q(t))-E(\mu(0),p(0),q(0))\mid\leq C\epsilon, \;\;\;\;\;\;for \;\;\;0\leq t\leq T=min(\frac{1}{\epsilon^2},T_0),
$$
where $C$ is a positive constant and time $T_0$ is the time which solution $\bar{\mu}_1$ of the averaged equation remains in $\mathfrak{D}$ .
\end{theorem}
\begin{proof}
 By averaging theorem \ref{t22}, we know that the accuracy of averaged equation (\ref{e28}) is of order $\epsilon$ over a time interval of order $\frac{1}{\epsilon^2}$. By theorem \ref{t210}, $E$ is an integral of the averaged system. So, $E=\frac 12 \langle\mu+\mathbb{I}V_0,\mathbb{I}^{-1}\mu+V_0\rangle$ is an adiabatic invariant for the perturbed Hamiltonian system studied in section 3.2.
\end{proof}

\section{Application: the Craik-Leibovich equation}

\subsection{Perturbation theory and an $n-$dimensional Craik-Leibovich equation}
\par
In this section, we show how to use the multiscale method discussed in section 4.1 to derive the Craik-Leibovich equation on any Riemannian manifold $D$ in any dimension.

\par
Let $D\subset \mathbb{R}^n$ be an $n-$dimensional manifold, and $M_{Emb}$ is an infinite-dimensional manifold of all the volume-preserving embeddings of the reference manifold $D$ into $\mathbb{R}^n$. This manifold $M_{Emb}$ is the configuration space for the inertia motion of fluid with free boundary. The group $SDiff(D)$ is the Lie group of volume-preserving diffeomorphisms of $D$. The infinite-dimensional manifold $N_{boun}$ is the manifold of all boundaries, where the boundaries are the images of maps in $M_{Emb}$ restricted to $\partial D$. There is a natural right action of the group $SDiff(D)$ on the configuration space $M_{Emb}$. So we have a principal $SDiff(D)-$ bundle $\pi:M_{Emb}\rightarrow N_{boun}$.

\par
As before, we introduce the space $\widetilde{M}_{Emb}=\pi^*(T^*N_{boun})$ and consider a natural Hamiltonian system on the reduced space $S_{free}=\widetilde{M}_{Emb}\times_{SDiff(D)} SVect(D)^*$. Locally, this space is isomorphic to $T^*U_{boun}\times SVect(D)^*$ for an open subset $U_{boun}$ of $N_{boun}$. In local coordinates, the Hamiltonian function is
$$
H(\Sigma,\;\phi,\;\mu)=\frac 12 \left\langle\mu,\mathbb{I}^{-1}\mu\right\rangle+\frac 12 \int_{D_{\Sigma}}(\nabla Hor(\phi), \nabla Hor(\phi)) dV+V(\Sigma),
$$
where, operator $Hor$ maps a function on $\Sigma$ to a function on the manifold $D_{\Sigma}$ bounded by $\Sigma$:
$$
\Phi=Hor(f) \;\;\text{such that}\;\;\Delta\Phi=0,\;\frac{\partial \Phi}{\partial n}\vert_{\Sigma}=f.
$$
(We denote this operator by $Hor$ because it actually corresponds to the horizontal lift on the principal bundle.)
So, vector field $\nabla Hor(f)$ is a gradient field on the manifold $D_{\Sigma}$. The above Hamiltonian describes the energy of an incompressible fluid with a free boundary. The first 2 terms constitute the kinetic energy. The last term is a potential energy related to the boundary $\Sigma$ (for instance, related to the surface tension).
\par
The equations of this Hamiltonian system are
\begin{equation}\label{e32}
\left\{
\begin{aligned}
\dot{\mu}&=\;-\mathcal{L}_{\;\mathbb{I}^{-1}\mu-\widetilde{A}_{free,\Sigma}\left(\nabla Hor\left(\phi\right)\right)}\;\mu,\\
\dot{\phi}&=\;-\nabla V(\Sigma)+\widetilde{\Omega}^*_{free,\Sigma,\phi}\;\mu+\widetilde{A}^*_{free,\Sigma}\mathcal{L}_{\mathbb{I}^{-1}\mu}\;\mu,\\
\dot{\Sigma}&=\;\phi.
\end{aligned}
\right.
\end{equation}
Here we use the operators $\widetilde{A}_{free,\Sigma}:T_{\Sigma} U_{boun}\rightarrow SVect(D)$ and $\widetilde{\Omega}_{free,\Sigma,\frac{\delta H}{\delta \phi}}=\widetilde{\Omega}_{free,\Sigma}\left(\nabla Hor\left(\frac{\delta H}{\delta \phi}\right),\nabla Hor\left(\cdot\right)\right):T_{\Sigma} U_{boun}\rightarrow SVect(D)$, as well as their duals $\widetilde{A}_{free,\Sigma}^*:SVect(D)^*\rightarrow T^*_{\Sigma} U_{boun}$ and $\widetilde{\Omega}_{free,\Sigma,\frac{\delta H}{\delta \phi}}^*:SVect(D)^*\rightarrow T^*_{\Sigma} U_{boun}$, whose geometric meanings are explained in the appendix.
\par
Look for the solution of the form $(\mu,\;\Sigma,\;\phi)(s,\tau)$, where we take the fast time $\tau=\frac{t}{\epsilon}$ and the slow time $s=\epsilon t$. Note that $\tau-$dependence is $2\pi-$periodic, but $s-$dependence is not necessarily so.
\par
As in section 4.1, the only nontrivial averaging is obtained from the first equation of (\ref{e32}). It becomes
\begin{equation}\label{e33}
\frac{1}{\epsilon}\frac{\partial}{\partial\tau}\mu+\epsilon\frac{\partial}{\partial s}\mu=\;-\mathcal{L}_{\mathbb{I}^{-1}\mu-\widetilde{A}_{free,\Sigma}\left(\nabla Hor\left(\phi\right)\right)}\;\mu.
\end{equation}
Expanding $\mu$ and $v$, we get
$$
\mu=\mu_0(s,\tau)+\epsilon\mu_1(s,\tau)+\epsilon^2\mu_2(s,\tau)+O(\epsilon^3),
$$
$$
v=v_0(s,\tau)+\epsilon v_1(s,\tau)+\epsilon^2 v_2(s,\tau)+O(\epsilon^3).
$$
Now split them into average and oscillating parts. Our three assumptions are as follows:
\par
1. $\bar{\mu}_0=0,\;\bar{v}_0=0$.
\par
2. $\tilde{v}_0$ does not depend on $s$. (Note that assumptions 1,2 mean that the zeroth approximation of the motion is potential flow, which is assumed in Craik-Leibovich's gravity wave model.)
\par
3. $\nabla V(\Sigma)|_{\Sigma=\overline{\Sigma}}=0$, which means that there is no inertia force at average position. \\
Applying the same two-timing method from section 4.1 we obtain the averaged equation
\begin{equation*}
 \frac{d}{ds}\;\bar{\mu}_1=-\mathcal{L}_{\bar{v}_1+V_0}\;\bar{\mu}_1,
\end{equation*}
where $\mu_1=\mathbb{I}\;v_1$, $V_0=\left\langle\frac 12[\tilde{v}_0,\tilde{v}_0^{\tau}]\right\rangle$. This is the \textbf{\textit{$n-$dimensional Craik-Leibovich (CL) equation}}.
\begin{theorem}\label{t31}
The $n-$dimensional CL equation is
\begin{equation}\label{e34}
 \frac{d}{ds}\;[u]=-\mathcal{L}_{v+V_0}\;[u],
\end{equation}
where $v+V_0\in SDiff(D)$, and $[u]=[v^b]$ is an element in the quotient $\Omega^1(D)/d\Omega^0(D)$, the regular dual to the Lie algebra $SVect(D)$.
\end{theorem}
\begin{proof}
The equation (\ref{e34}) is
$$
\frac{d}{ds}\;u=-\mathcal{L}_{v+V_0}\;u+d\psi,
$$
or, equivalently,
$$
\frac{d}{ds}\;u=-\mathcal{L}_{v}\;u-\mathcal{L}_{V_0}\;u+d\psi.
$$
By using the identities
$$
\mathcal{L}_v(v^b)=(\nabla_v v)^b+\frac 12 d\langle v,\;v\rangle
$$
and
$$
\mathbb{I}(*(curl \;v\wedge V_0))=[i_{V_0}i_{curl \;v}\mu]=[i_{V_0}dv^b]=\mathcal{L}_{V_0}[u],
$$
we obtain an $n-$dimensional Craik-Leibovich equation
\begin{equation}\label{e35}
 v_t+\nabla_v v+curl \;v\times V_0=-\nabla p,
\end{equation}
where $curl\;v$ is defined as an $(n-2)-$vector field.
\end{proof}
\begin{remark}
 According to the theorem \ref{t22}, we know that the difference between the solution of the Craik-Leibovich equation (i.e. the averaged equation) and the solution of the perturbed Euler equation remains small (of order $\epsilon$) over a very long time interval (of order $\frac{1}{\epsilon^2}$).
\end{remark}

The CL equation (\ref{e34}) turns out to be a Hamiltonian equation,
\begin{corollary}\label{t32}
Equation (\ref{e34}) is a Hamiltonian equation on coadjoint orbits in $\mathfrak{g}^*=\Omega^1(D)/d\Omega^0(D)$ with the Hamiltonian function $H=-\frac 12([u+V_0^b],\;[u+V_0^b])$.
\end{corollary}
\begin{proof}
By proposition (\ref{p11}), the Hamiltonian equation for the function $H=\frac 12([u+V_0^b],\;[u+V_0^b])$ is
$$
\frac{d}{ds}\;[u]=-\mathcal{L}_{\mathbb{I}^{-1}[u+V_0^b]}\;[u],
$$
which is the equation (\ref{e34})
\end{proof}
\begin{theorem}\label{t33}
The equation (\ref{e34}) is the Euler equation on the central extension of the Lie algebra $\mathfrak{g}=SVect(D)$ by means of the 2-cocycle
$$
\widehat{\omega}_{V_0}(X,Y)=-\left\langle\mathcal{L}_{V_0}\;X^b,Y\right\rangle
$$
associated to the vector field $V_0$.
\end{theorem}
\begin{proof}
This follows from theorem \ref{t11}.
\end{proof}
\begin{remark}
 This is the Lichnerowicz 2-cocycle for the 2-form $\beta=-dV_0^b$.
\end{remark}

From the geometric structure of the CL equation, one can derive first integrals for the CL equation. These invariants are studied in \cite{khch} in the similar setting of the infinite conductivity equation,
\begin{corollary}\label{t34}
Equation (\ref{e34}) has
\par
(1) an integral $I(v)=\int_D u\wedge\;(du)^m$ for $u=v^b$ in the case of an odd $n=2m+1$,
\par
(2) infinitely many integrals
$$
I_f(v)=\int_D f\left(\frac{(du)^m}{vol_D}\right)\;vol_D
$$
in the case of an even $n=2m$, here $vol_D$ is the volume form on $D$.
\end{corollary}
\begin{proof}
Let $SDiff(D)$ be the group of volume-preserving diffeomorphisms on $D$. The form of equation (\ref{e34}) shows that the moment $[u]$ moves along coadjoint orbits of the $SDiff(D)$-action corresponding to $v+V_0$. The functionals (1) and (2) are constant on the coadjoint orbits since $SDiff(D)$-action coincides with the change of variables and preserves the volume form on $D$.
\end{proof}

\subsection{Adiabatic invariants}
Now consider the perturbation model discussed in section 3.3. Assume that the velocity field $v=\mathbb{I}^{-1}\mu-\widetilde{A}_{free,\Sigma}\left(\nabla Hor\left(\phi\right)\right)$ has the form
$$
v=\epsilon v_1+\epsilon^2 v_2+\epsilon^3 v_3+\dots,
$$
where $v_1$ is a velocity field such that $v_1=-\widetilde{A}_{free,\Sigma}\left(\nabla Hor\left(\phi\right)\right)$ and periodic with respect to time $t$.
Also for $\mu$ we have
$$
\mu=\epsilon^2 \mu_2+\epsilon^3 \mu_3+\epsilon^4 \mu_4+\dots.
$$
Now consider the slow time $s=\epsilon^2 t$, and let $v=v(s,t)$ and $\mu=\mu(s,t)$. Split them into the average and oscillating parts,
$$
v(s,t)=\bar{v}(s)+\tilde{v}(s,t),
$$
$$
\mu(s,t)=\bar{\mu}(s)+\tilde{\mu}(s,t).
$$
Applying the same argument as in section 3.3, we can get the averaged equation
$$
\frac{d}{dt}\;\bar{\mu}_2=-\epsilon^2\mathcal{L}_{\bar{v}_2+V_1}\;\bar{\mu}_2,
$$
where $\mu_2=[v_2^b]$ is an element in the dual $\Omega^1(D)/d\Omega^0(D)$.

\par
Theorem \ref{t24} implies the following
\begin{theorem}\label{t35}
 The functional $H=-\frac 12(\mu+[V_0^b],\;\mu+[V_0^b])$ is an adiabatic invariant for the perturbed Hamiltonian system described above. Namely, let $\mathfrak{D}$ be a bounded domain in the energy norm, then there exists $\epsilon_0$ such that for all $\epsilon$ satisfying $0<\epsilon<\epsilon_0$, one has
$$
\mid H(\mu(t),\phi(t),\Sigma(t))-H(\mu(0),\phi(0),\Sigma(0))\mid\leq C\epsilon \;\;\;\;\;\;for \;\;\;0\leq t\leq \frac{1}{\epsilon^2},
$$
where $C$ is a positive constant and time $T_0$ is the time during which the solution $\bar{\mu}_2$ of the averaged equation remains in $\mathfrak{D}$ .
\end{theorem}
\par
In addition to this adiabatic invariant, there is a large class of adiabatic invariants for this perturbation model:

\begin{theorem}\label{t36}

For the perturbed Hamiltonian system discussed above,
\par
 (1) for $n=2m+1$, the functional $I(v)=\int_D u\wedge\;(du)^m$ is an adiabatic invariant for 1-form $u=v^b$;
\par
 (2) for $n=2m$, the functionals $I_f(v)=\int_D f\left(\frac{(du)^m}{vol_D}\right)\;vol_D $ are adiabatic invariants for any function $f:\mathbb{R}\rightarrow\mathbb{R}$ and 1-form $u=v^b$.
\end{theorem}
\begin{proof}
 By averaging theorem \ref{t22}, we know that the accuracy of the averaged equation (\ref{e28}) is of order $\epsilon$ over a time interval of order $\frac{1}{\epsilon^2}$. Now, the conclusion follows from theorem \ref{t34}.
\end{proof}

\section{Appendix: The geometry and Hamiltonian structure on the reduced space of a principal $G-$bundle}
\subsection{Local Poisson structures on reduced spaces}
\par
In order to study the Hamiltonian structure on the reduced space of a principal $G-$bundle, we use the local Poisson bracket studied in \cite{lew}, \cite{mont}.
\par
Let $\pi:M\rightarrow N$ be a principal $G-$bundle, $N=M/G$ is its base space. We are going to describe the local Poisson structure on the reduced space $T^*M/G$. This Poisson structure is actually defined on an associated bundle isomorphic to $T^*M/G$. First, we explain the construction of this associated bundle.
\par
Let $\widetilde{M}=\pi^*(T^*N)$ be a principle $G-$bundle $\widetilde{\pi}:\widetilde{M}\rightarrow T^*N$. This is a subbundle of the cotangent bundle $T^*M$. Now consider the associated bundle $S=\widetilde{M}\times_{G} \mathfrak{g}^*$, which consists the orbits of $G$ action on $\widetilde{M}\times \mathfrak{g}^*$ given by
$$
(\alpha_q,\mu)\circ g=(TR^*_{g^{-1}}\alpha_q,Ad^{*}_g\mu), \;g\in G, \;\alpha_q\in\widetilde{M},\;\mu\in\mathfrak{g}^*.
$$
This space is studied in \cite{ster}.
\par
The isomorphism between $S=\widetilde{M}\times_{G} \mathfrak{g}^*$ and $T^*M/G$ is induced by the isomorphism between $\widetilde{M}\times \mathfrak{g}^*$ and $T^*M$:
\begin{align*}
\widetilde{M}\times \mathfrak{g}^*  &\longrightarrow  T^*M\\
(\alpha_q,\mu)  &\longmapsto  \alpha_q+A^*_q\mu,
\end{align*}
where $A\in\Omega^1(M,g)$ is the connection 1-form on $M$.
\par
Let us consider a local section $\gamma:U\rightarrow M$ for an open subset $U$ of $N$. Then, locally,

$$
\widetilde{M}\times \mathfrak{g}^*\simeq (T^*U\times G)\times \mathfrak{g}^*,
$$
$$
S=\widetilde{M}\times_{G} \mathfrak{g}^*\simeq T^*U\times \mathfrak{g}^*.
$$
So, in local coordinates, elements in $S$ can be written as $(q,\;p,\;\mu)$, here $(q,\;p)\in T^*U$, $\mu \in \mathfrak{g}^*$ and the isomorphism $S\simeq T^*M/G$ is given by
$$
(q,\;p,\;\mu)\longmapsto (x,\;p+\widetilde{A}^*_q\mu),
$$
where $\widetilde{A}\in \Omega^1(N,\mathfrak{g})$ is a 1-form induced by the local trivialization $\widetilde{A}_q=\gamma^*A_{\gamma(q)}$.
\par
The local formula for the Poisson structure on $S$ is
\begin{align*}
 \{F,\;G\}(q,\;p,\;\mu)&=\frac{\delta F}{\delta q}\frac{\delta G}{\delta p}-\frac{\delta G}{\delta q}\frac{\delta F}{\delta p}+\left\langle\mu,\;-\left[\widetilde{A}_q\frac{\delta F}{\delta p},\frac{\delta G}{\delta \mu}\right]+\left[\widetilde{A}_q\frac{\delta G}{\delta p},\frac{\delta F}{\delta \mu}\right]\right\rangle\\
&+\left\langle\mu,\;\widetilde{\Omega}_q\left(\frac{\delta F}{\delta p},\frac{\delta G}{\delta p}\right)\right\rangle+\left\langle\mu,\;\left[\frac{\delta F}{\delta \mu},\frac{\delta G}{\delta \mu}\right]\right\rangle,
\end{align*}
where $\Omega$ is the curvature corresponding to the connection $A$ and $\widetilde{\Omega}=\gamma^*\Omega\in \Omega^2(N,\mathfrak{g})$.
\subsection{Hamiltonian equations on the reduced space}

\par
Let $\pi:M\rightarrow N$ be a principle $G-$bundle, $U$ is an open subset of $N$. Consider a local Hamiltonian function $H(q,p,\mu)$ on $T^*U\times \mathfrak{g}^*$. We derive the Hamiltonian equation corresponding to $H$.
\par
In the following proposition $A\in\Omega^1(M,\mathfrak{g})$ stands for the connection 1-form and $\widetilde{A}\in\Omega^1(N,\mathfrak{g})$ is a 1-form induced by a local section: $\gamma:U\rightarrow M$. We have the operators $\widetilde{A}_q:T_q U\rightarrow \mathfrak{g}$ and $\widetilde{\Omega}_{q,\frac{\delta H}{\delta p}}=\widetilde{\Omega}_q\left(\frac{\delta H}{\delta p},\cdot\right):T_q U\rightarrow \mathfrak{g}$, as well as their duals $\widetilde{A}_q^*:\mathfrak{g}^*\rightarrow T^*_q U$ and $\widetilde{\Omega}_{q,\frac{\delta H}{\delta p}}^*:\mathfrak{g}^*\rightarrow T^*_q U$.
\begin{proposition}\label{p21}
The Hamiltonian equations of $H(q,\;p,\;\mu)$ in local coordinates are
\begin{equation}\label{e21}
\left\{
\begin{aligned}
\dot{\mu}&=\;-ad^*_{\frac{\delta H}{\delta \mu}-\widetilde{A}_q\frac{\delta H}{\delta p}}\;\mu,\\
\dot{p}&=\;-\frac{\delta H}{\delta q}+\widetilde{A}^*_q ad^*_{\frac{\delta H}{\delta \mu}}\;\mu-\widetilde{\Omega}_{q,\frac{\delta H}{\delta p}}^*\;\mu,\\
\dot{q}&=\;\frac{\delta H}{\delta p}.
\end{aligned}
\right.
\end{equation}

\end{proposition}
\begin{proof}
 We have
 \begin{align*}
 \{F,\;H\}(q,\;p,\;\mu)&=\frac{\delta F}{\delta q}\frac{\delta H}{\delta p}-\frac{\delta H}{\delta q}\frac{\delta F}{\delta p}+\left\langle\mu,\;-\left[\widetilde{A}_q\frac{\delta F}{\delta p},\frac{\delta H}{\delta \mu}\right]+\left[\widetilde{A}_q\frac{\delta H}{\delta p},\frac{\delta F}{\delta \mu}\right]\right\rangle\\
 &+\left\langle\mu,\;\widetilde{\Omega}_q\left(\frac{\delta F}{\delta p},\frac{\delta H}{\delta p}\right)\right\rangle+\left\langle\mu,\;\left[\frac{\delta F}{\delta \mu},\frac{\delta H}{\delta \mu}\right]\right\rangle\\
 &=\frac{\delta F}{\delta q}\frac{\delta H}{\delta p}-\frac{\delta H}{\delta q}\frac{\delta F}{\delta p}+\left\langle\widetilde{A}^*_qad^*_{\frac{\delta H}{\delta \mu}}\mu,\;\frac{\delta F}{\delta p}\right\rangle+\left\langle ad^*_{\widetilde{A}_q\frac{\delta H}{\delta \mu}}\mu,\;\frac{\delta F}{\delta \mu}\right\rangle\\
 &-\left\langle\widetilde{\Omega}_{q,\frac{\delta H}{\delta p}}^*\;\mu,\;\frac{\delta F}{\delta p}\right\rangle-\left\langle ad^*_{\frac{\delta H}{\delta \mu}}\mu,\;\frac{\delta F}{\delta \mu}\right\rangle.
\end{align*}
Comparing this with
$$
\dot{F}=\frac{\delta F}{\delta \mu}\dot{\mu}+\frac{\delta F}{\delta p}\dot{p}+\frac{\delta F}{\delta q}\dot{q},
$$
we obtain Hamiltonian equations (\ref{e21}).
\end{proof}
\par
Now let us consider a natural Hamiltonian system on the reduced space $S=\widetilde{M}\times_G \mathfrak{g}^*$. In local coordinates the Hamiltonian function is $H(q,\;p,\;\mu)=\frac 12 \|p\|^2+\frac 12 \langle\mu,\mathbb{I}^{-1}\mu\rangle+V(q)$, where $\|\cdot\|$ is the norm on $T^*U$ induced by the Riemannian metric.
\par
By proposition \ref{p21}, the Hamiltonian equations of the natural Hamiltonian system are
\begin{equation}
\left\{
\begin{aligned}
\dot{\mu}&=\;-ad^*_{\mathbb{I}^{-1}\mu-\widetilde{A}_q\frac{\delta H}{\delta p}}\;\mu,\\
\dot{p}&=\;-\nabla V(q)+\widetilde{A}^*_q ad^*_{\mathbb{I}^{-1}\mu}\;\mu-\widetilde{\Omega}_{q,\frac{\delta H}{\delta p}}^*\;\mu,\\
\dot{q}&=\;\frac{\delta H}{\delta p},
\end{aligned}
\right.
\end{equation}
where $\frac{\delta H}{\delta p}\in TU$ is a velocity field on the open subset $U$.
\subsection{Geometry of an incompressible fluid with a free boundary}
\par
Let $D\subset \mathbb{R}^n$ be an $n-$dimensional manifold, and $M_{Emb}$ is an infinite-dimensional manifold that contains all the volume-preserving embeddings of the reference manifold $D$ into $\mathbb{R}^n$. The group $SDiff(D)$ is the Lie group of volume-preserving diffeomorphisms of $D$. An infinite-dimensional manifold $N_{boun}$ is the manifold of all boundaries, where the boundaries are the images of maps in $M_{Emb}$ restricted to $\partial D$. There is a natural right action of $SDiff(D)$ on $M_{Emb}$. So we have a principal $SDiff(D)-$ bundle $\pi:M_{Emb}\rightarrow N_{boun}$.

\par
Then we consider the reduced space $S_{free}=\widetilde{M}_{Emb}\times_{SDiff(D)} SVect(D)^*$, where the bundle $\widetilde{M}_{Emb}=\pi^*(T^*N_{boun})\subset T^*M_{Emb}$. Here $SVect(D)$ is the Lie algebra of $SDiff(D)$ and $SVect(D)^*=\Omega^1(D)/d\Omega^0(D)$ is the dual of the Lie algebra $SVect(D)$. Elements in the cotangent bundle $T^*N_{boun}$ are $(\Sigma,\phi)$ where $\Sigma$ denotes the boundary of a manifold $D_{\Sigma}\subset \mathbb{R}^n$, and $\phi$ is a function on $\Sigma$ such that $\int_{\Sigma}\phi\;vol_{\Sigma}=0$.

\par
Now, we are going to define a connection on this principal bundle $M_{Emb}$. First, we introduce an operator $Hor$ which maps a function on $\Sigma$ to a function on the manifold $D_{\Sigma}$ bounded by $\Sigma$:
$$
\Phi=Hor(f) \;\;\text{such that}\;\;\Delta\Phi=0,\;\frac{\partial \Phi}{\partial n}\vert_{\Sigma}=f.
$$
So, vector field $\nabla Hor(f)$ is a gradient field on $D_{\Sigma}$. Given a vector field $v$ on $D_{\Sigma}$, the gradient $\nabla Hor((v,n))$ is its horizontal component. This gives us the connection $A_{free}\in\Omega^1(M_{Emb},SVect(D)^*)$ on the space $M_{Emb}$ of embeddings.

\par
Let $U_{boun}$ be an open subset of $N_{boun}$, and suppose we have a local section $\gamma_{free}:U_{boun}\rightarrow M_{Emb}$. Locally, $\widetilde{M}_{Emb}\times SVect(D)^*\simeq(T^*U_{boun}\times SDiff(D))\times SVect(D)^*$, so the reduced space $S_{free}=\widetilde{M}_{Emb}\times_{SDiff(D)} SVect(D)^*\simeq T^*U_{boun}\times SVect(D)^*$. Also, $\widetilde{A}_{free}\in\Omega^1(U_{boun},SVect(D))$ is a 1-form induced by the local trivialization $\widetilde{A}_{free,\Sigma}=\gamma^*_{free}A_{free,\gamma_{free}(\Sigma)}$
\par
In local coordinates, elements of $S$ are $(\Sigma,\phi,\mu)$, where $(\Sigma,\phi)\in T^*N$ and $\mu\in \mathfrak{g}^*$. The Poisson structure in local coordinates is
\begin{align*}
 &\{F,\;G\}(\Sigma,\;\phi,\;\mu)=\int_D\left(\left\langle\mu,\;\left[\frac{\delta F}{\delta \mu},\frac{\delta G}{\delta \mu}\right]\right\rangle+\left\langle\mu,\;\widetilde{\Omega}_{free,\Sigma}\left(\nabla Hor\left(\frac{\delta F}{\delta \phi}\right),\nabla Hor\left(\frac{\delta G}{\delta \phi}\right)\right)\right\rangle\right)\;vol_D\\
                               &+\int_D\left(\left\langle\frac{\delta F}{\delta \mu},\mathcal{L}_{\widetilde{A}_{free,\Sigma}\left(\nabla Hor\left(\frac{\delta G}{\delta \phi}\right)\right)}\;\mu\right\rangle-\left\langle\frac{\delta G}{\delta \mu},\mathcal{L}_{\widetilde{A}_{free,\Sigma}\left(\nabla Hor\left(\frac{\delta F}{\delta \phi}\right)\right)}\;\mu\right\rangle\right)\;vol_D+\int_{\Sigma}\left(\frac{\delta F}{\delta \Sigma}\frac{\delta G}{\delta \phi}-\frac{\delta G}{\delta \Sigma}\frac{\delta F}{\delta \phi}\right)\;vol_{\Sigma}
\end{align*}
where, $\widetilde{\Omega}_{free}=\gamma^*_{free}\Omega_{free}\in \Omega^2(N_{boun},SVect(D))$ and $\Omega_{free}$ is the curvature corresponding to the connection $A_{free}$.
\begin{remark}
To define $\frac{\delta F}{\delta \Sigma}$ in the above formula, first we need to fix the vector field $v$ on the manifold $D_{\Sigma}$ bounded by $\Sigma$, the horizontal component of $v$ is $\nabla Hor((v,n))$ and the vertical component of $v$ is $v_{\|}=v-\nabla Hor((v,n))$, which is a divergence-free vector field tangent to $D_{\Sigma}$. Then, extend $v$ smoothly to the neighbourhood of $\Sigma$ such that one can keep $v$ fixed while varying $\Sigma$. Finally, the functional derivative $\frac{\delta F}{\delta \Sigma}$ is given by
$$
D_{\Sigma}F(v,\Sigma)\cdot\delta\Sigma=\int_{\Sigma}\frac{\delta F}{\delta \Sigma}\delta\Sigma\;vol_{\Sigma}.
$$
\end{remark}
\subsection{The Hamiltonian equation on the reduced space for the free boundary problem}
In the following proposition, $A_{free}\in\Omega^1(M_{Emb},SVect(D))$ is the connection 1-form and $\widetilde{A}_{free}\in\Omega^1(N_{boun},SVect(D))$ is a 1-form induced by a local section: $\gamma_{free}:U_{boun}\rightarrow M_{Emb}$. We have the operators $\widetilde{A}_{free,\Sigma}:T_{\Sigma} U_{boun}\rightarrow SVect(D)$ and $\widetilde{\Omega}_{free,\Sigma,\frac{\delta H}{\delta \phi}}=\widetilde{\Omega}_{free,\Sigma}\left(\nabla Hor\left(\frac{\delta H}{\delta \phi}\right),\nabla Hor\left(\cdot\right)\right):T_{\Sigma} U_{boun}\rightarrow SVect(D)$, as well as their duals $\widetilde{A}_{free,\Sigma}^*:SVect(D)^*\rightarrow T^*_{\Sigma} U_{boun}$ and $\widetilde{\Omega}_{free,\Sigma,\frac{\delta H}{\delta \phi}}^*:SVect(D)^*\rightarrow T^*_{\Sigma} U_{boun}$.
Apply proposition \ref{p21} to this model to obtain:
\begin{proposition}\label{p31}
 The Hamiltonian equations for the Hamiltonian function $H(\Sigma,\;\phi,\;\mu)$ in local coordinate are
\begin{equation}\label{e31}
\left\{
\begin{aligned}
\dot{\mu}&=\;-\mathcal{L}_{\frac{\delta H}{\delta \mu}-\widetilde{A}_{free,\Sigma}\left(\nabla Hor\left(\frac{\delta H}{\delta \phi}\right)\right)}\;\mu,\\
\dot{\phi}&=\;-\frac{\delta H}{\delta \Sigma}+\widetilde{\Omega}^*_{free,\Sigma,\frac{\delta H}{\delta \phi}}\;\mu+\widetilde{A}^*_{free,\Sigma}\mathcal{L}_{\frac{\delta H}{\delta \mu}}\;\mu,\\
\dot{\Sigma}&=\;\frac{\delta H}{\delta \phi}.
\end{aligned}
\right.
\end{equation}
\end{proposition}
\begin{proof}
We have
\begin{align*}
&\{F,\;H\}(\Sigma,\;\phi,\;\mu)=\int_D\left(-\left\langle\mathcal{L}_{\frac{\delta H}{\delta \mu}}\;\mu,\frac{\delta F}{\delta \mu}\right\rangle+\left\langle\mathcal{L}_{\widetilde{A}_{free,\Sigma}\left(\nabla Hor\left(\frac{\delta H}{\delta \phi}\right)\right)}\;\mu,\frac{\delta F}{\delta \mu}\right\rangle\right)\;vol_D\\
&+\left\langle\widetilde{\Omega}^*_{free,\Sigma,\frac{\delta H}{\delta \phi}}\;\mu, \frac{\delta F}{\delta \phi}\right\rangle+\left\langle\widetilde{A}^*_{free,\Sigma}\mathcal{L}_{\frac{\delta H}{\delta \mu}}\;\mu, \frac{\delta F}{\delta \phi}\right\rangle +\int_{\Sigma}\left(\frac{\delta F}{\delta \Sigma}\frac{\delta H}{\delta \phi}-\frac{\delta H}{\delta \Sigma}\frac{\delta F}{\delta \phi}\right)\;vol_{\Sigma}
\end{align*}
\par
Comparing this with
$$
\dot{F}=\frac{\delta F}{\delta \mu}\dot{\mu}+\frac{\delta F}{\delta \Sigma}\dot{\Sigma}+\frac{\delta F}{\delta \phi}\dot{\phi},
$$
we obtain the system (\ref{e31}).
\end{proof}
\par
Now, we consider the Hamiltonian function
$$
H(\Sigma,\;\phi,\;\mu)=\frac 12 \left\langle\mu,\mathbb{I}^{-1}\mu\right\rangle+\frac 12 \int_{D_{\Sigma}}(\nabla Hor(\phi), \nabla Hor(\phi))\;vol_{D_{\Sigma}}+V(\Sigma).
$$
This Hamiltonian describes the energy of an incompressible fluid with a free boundary. The first two terms constitute the kinetic energy. The last term is a potential energy related to the boundary $\Sigma$. For a divergence-free vector field $v$ on $D_{\Sigma}$, which is not necessarily tangent to $\Sigma$, by the Hodge decomposition theorem, one has
$$
v=v_{\|}+\nabla Hor(\phi),
$$
where $v_{\|}$ is tangent to $\Sigma$ and $\phi=v\cdot n$ for a field of exterior normals $n$ to $\Sigma$.
The Hamiltonian equations of this Hamiltonian system are
\begin{equation}
\left\{
\begin{aligned}
\dot{\mu}&=\;-\mathcal{L}_{\;\mathbb{I}^{-1}\mu-\widetilde{A}_{free,\Sigma}\left(\nabla Hor\left(\phi\right)\right)}\;\mu,\\
\dot{\phi}&=\;-\nabla V(\Sigma)+\widetilde{\Omega}^*_{free,\Sigma,\phi}\;\mu+\widetilde{A}^*_{free,\Sigma}\mathcal{L}_{\mathbb{I}^{-1}\mu}\;\mu,\\
\dot{\Sigma}&=\;\phi.
\end{aligned}
\right.
\end{equation}

\bigskip

{\bf  Acknowledgements.}
The author is grateful to Boris Khesin for suggesting the problems and for many fruitful discussions.
This research is partially supported by the China Scholarship Council.

  \end{document}